\documentclass{article}
\usepackage{graphicx}
\usepackage{enumerate}
\usepackage{caption}
\usepackage{subcaption}
\usepackage{algorithm}
\usepackage{algpseudocode}
\usepackage{amsmath}
\usepackage{url}
\usepackage{amssymb}
\usepackage{color}
\newtheorem{defini}{Definition}

\newtheorem{theorem}{Theorem}

\newtheorem{lemma}[theorem]{Lemma}
\newtheorem{cor}[theorem]{Corollary}
\newtheorem{obs}{Observation}

\def\QED{\ensuremath{{\square}}}
\def\markatright#1{\leavevmode\unskip\nobreak\quad\hspace*{\fill}{#1}}
\newenvironment{proof}
  {\begin{trivlist}\item[\hskip\labelsep{\bf Proof.}]}
  {\markatright{\QED}\end{trivlist}}

\title{A Fast 2-Approximation Algorithm for Guarding Orthogonal Terrains }

\author{Yangdi Lyu\thanks{Dept. of Computer \& Info. Sci. \& Eng.,
        University of Florida, {\tt \{yangdi, ungor\}@cise.ufl.edu}}
        \and
        Alper \"Ung\"or\footnotemark[1]}
\begin{document}
\maketitle

\begin{abstract}
Terrain Guarding Problem(TGP), which is known to be NP-complete, asks to find a smallest set of guard locations on a terrain $T$ such that every point on $T$ is visible by a guard. Here, we study this problem on 1.5D orthogonal terrains where the edges are bound to be horizontal or vertical. We propose a 2-approximation algorithm that runs in O($n \log m$) time, where $n$ and $m$ are the sizes of input and output, respectively. This is an improvement over the previous best algorithm, which is a 2-approximation with O($n^2$) running time.

\end{abstract}

\section{Introduction}
Optimal placement of antennas, cameras, and light sources on terrains is important for communication network, security, and architectural design applications. Even a consideration of the problem on 1.5D terrains is useful whenever the domain is a highway, street, or a hallway. Moreover, this simpler version plays a role on the complexity analysis and algorithm design for the guarding problem on higher dimensional terrains. 

A 1.5D terrain $T$ is an $x$-monotone polygonal chain consists of $n$ vertices $v_i \in \mathbb{R}^2$, for $ i=1, 2, \ldots,n$ and $n-1$ edges $e_{i} = \overline{v_{i}v_{i+1}}$ for $i=1,2,\ldots,n-1$.
$T$ is called an \emph{orthogonal terrain} if all its edges are either horizontal or vertical, and there are no two consecutive horizontal/vertical edges. 
For two vertices $p, q \in T$, we say $p$ is left of $q$, denoted as $p<q$, if $p.x < q.x$. The vertices of $T$ are indexed from left to right, so $v_{i+1} \nless v_i$. For $p, q \in T$, $p$ can \textit{see} $q$ if the line segment $\overline{pq}$ is never strictly below the terrain $T$.

Given a terrain $T$, a guarding candidate set $G\subseteq{T}$ and a witness set $W\subseteq{T}$, terrain guarding problem TGP$(G,W)$ is to find the minimum guarding set $G^{*}\subseteq{G}$ such that each point in $W$ is seen by at least one point in $G^{*}$. For orthogonal terrains, we refer to this problem as OTGP. Here, we focus on solving OTGP$(V(T), V(T))$ where both the guarding candidate set and the witness set are the vertices of the terrain, i.e., $G=W=V(T)$.


\subsection{Related Work}
The terrain guarding problem is closely related to the well known Art Gallery Problem \cite{Rourke} of finding the minimum set of positions to guard a polygon. The first result was obtained by Chv\'atal: $\lfloor\frac{n}{3}\rfloor$ guards are always sufficient and sometimes necessary to guard a polygon of n vertices. Art Gallery Problem was shown to be NP-hard: on simple polygons \cite{Lee}, on simple orthogonal polygons \cite{Schuchardt}, and on monotone polygons \cite{Krohn}. Moreover, it was shown to be APX-hard on simple polygons \cite{Eidenbenz}.

Terrain Guarding Problem for general 1.5D terrains is shown to be NP-hard by a reduction from \textsc{Planar 3sat} \cite{King}. Ben-Moshe \textit{et al}. \cite{Moshe} gave the first $O(1)$-approximation algorithm. Elbassioni \textit{et al}. \cite{Elbassioni} gave an improvement by showing that LP rounding results in a 4-approximation for TGP$(G,W)$ if $G\cap{W}=\emptyset$ (a 5-approximation otherwise). A local search based PTAS is also proposed for TGP  \cite{Friedrichs,Gibson}.

For orthogonal terrains, Katz and Roisman \cite{Katz} gave a 2-approximation algorithm that runs in O($n^{2}$) time, by computing a minimum clique cover in chordal graphs. Recently, Durocher \textit{et al}. \cite{Durocher} studied the orthogonal terrain guarding problem under \emph{directed visibility} where two vertices $u, v$ are considered to see each other only if the interior of the segment $uv$ is strictly above the terrain.
Under this restricted definition, no reflex vertex of the input terrain $T$ can see convex vertices both on its left and right side. This property simplifies the problem, and leads to a linear time greedy exact algorithm. Under standard visibility, Durocher \textit{et al}. \cite{Durocher} also observed that the hardness result for TGP in \cite{King} does not apply for orthogonal terrains, leaving the complexity of OTGP open.

\section{Preliminaries}
We assume that the input terrain begins and ends with vertical edges, an assumption for technical convenience to be removed in the appendix.

$V(T)$ is split into two disjoint subsets as reflex vertices $V_{r}(T)$ and convex vertices $V_{c}(T)$. Walking along the orthogonal terrain $T$ from left to right, a vertex $v$ is \textit{convex(reflex)} if we turn left(right) at $v$. Each subset is further split into two subsets depending on whether a vertex is on the left or on the right side of its incident horizontal edge. Specifically, walking along $T$ from left to right, a vertex $v$ is \textit{left(right)} if we walk from a vertical(horizontal) edge to a horizontal(vertical) edge at $v$. So, $V(T)$ is split into four disjoint subsets: left reflex vertices $V_{lr}(T)$, right reflex vertices $V_{rr}(T)$, left convex vertices $V_{lc}(T)$, and right convex vertices $V_{rc}(T)$, see Figure~\ref{figdefinition}. The first and the last vertices of $T$ can also be labelled simply by considering dummy horizontal edges incident to them. 

For each $v\in{V_c(T)}$, \textit{upper vertex} of $v$, $U(v)\in{V_r(T)}$ is the reflex vertex that shares a common vertical edge with $v$, see Figure~\ref{figdefinition}. As $T$ begins and ends with vertical edges, $U(v)$ for each convex vertex $v$ is well defined.

For each $v\in{V_{lc}(T)}$, \textit{right horizon} of $v$, $R(v)\in{V_r(T)}$ is the rightmost reflex vertex that can see $v$, see Figure~\ref{figdefinition}. This definition is similar to that of $R(v)$ by Durocher \textit{et al}. \cite{Durocher} except that a left convex vertex cannot be seen by right reflex vertices under directed visibility but it can be seen by them under standard visibility. 

\begin{figure}[h]
\centering
\includegraphics[width =.9\linewidth]{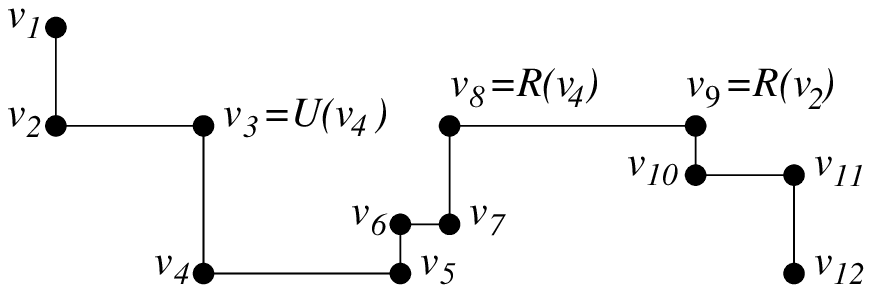}
\caption{$V_{lc}(T)$ = $\lbrace{v_{2}, v_{4}, v_{10}, v_{12}}\rbrace$, $V_{rc}(T)$ = $\lbrace{v_{5}, v_{7}}\rbrace$, $V_{lr}(T)$ = $\lbrace{v_{6}, v_{8}}\rbrace$, $V_{rr}(T)$ = $\lbrace{v_{1}, v_{3}, v_{9}, v_{11}}\rbrace$}
\label{figdefinition}
\end{figure}

Following definition by L\"offler \textit{et al.}~will also be used. 

\begin{defini}\emph{\cite{Loffler}}
Given a reflex vertex $p_i$ and a vertex $v_k \in T$, the ray with origin $p_i$ and vector $\overrightarrow{p_iv_k}$ is called a \emph{shadow ray} if: (i) $p_i$ sees $v_k$; (ii) $p_i$ does not see the points of $T$ immediately to the left of $v_k$.
\end{defini}

For each shadow ray $\overrightarrow{p_iv_k}$, $v_k$ is called the \emph{obstacle} of $p_i$, $obs(p_i)$. By definition, there may be multiple shadow rays for each vertex $p$, corresponding to different obstacles. The sweepline algorithm presented in the next Section relies on the following definition to identify a unique shadow ray (and its obstacle). The shadow ray of $p$ with respect to the sweep line at event $w$, $sr_w(p)$, is defined as the highest shadow ray of $p$ whose obstacle is to the right of the sweep line at event $w$, see Figure~\ref{figshadowray}. In the following sections, a shadow ray of $p$ refers to the shadow ray of $p$ with respect to the current sweep line. In our algorithm, lower envelope of shadow rays is maintained to extract some essential visibility information efficiently.
\begin{figure}[h]
\centering
\includegraphics[width = 0.7\linewidth]{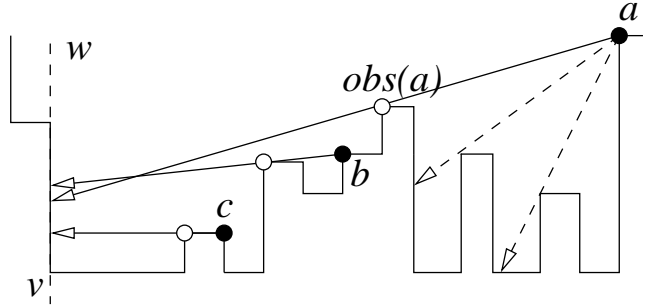}
\caption{The shadow rays of $a$, $b$ and $c$ with respect to sweep line $w$. Obstacles are denoted by empty circles.}
\label{figshadowray}
\end{figure}

\subsection{Properties of Orthogonal Terrains}
The following claim called \textit{the order claim} was proved by Ben-Moshe \textit{et al}. \cite{Moshe}, and holds in 1.5D general terrains.
\begin{lemma}\emph{\cite{Moshe}}
Let $p < q < r < s$ be four points on terrain $T$. If $p$ sees $r$, and $q$ sees $s$, then $p$ sees $s$.
\label{lemmaOrder}
\end{lemma}
The following claims were proved by Katz and Roisman \cite{Katz} for orthogonal terrains. 
\begin{lemma}\emph{\cite{Katz}}
Let $T$ be an orthogonal terrain, $v\in{V_{lc}(T)}$ and another point $p$ on $T$ can see $v$, then $p \nless v$.
\label{lemmaConvex}
\end{lemma}
\begin{lemma}\emph{\cite{Katz}}
If a set $G$ of points on orthogonal terrain $T$ guards a subset $V'\subseteq{V_{c}(T)}$, then there exists a subset $G'\subseteq{V_{r}(T)}$, such that $G'$ guards $V'$ and ${\vert{G'}\vert}\leqslant{\vert{G}\vert}$.
\label{lemmaReflexSub}
\end{lemma}
\begin{lemma}\emph{\cite{Katz}}
If $G\subseteq{V(T)}$ guards all the convex vertices of an orthogonal terrain $T$ (i.e., $G$ guards the set $V_{c}(T)$), then $G$ guards all the vertices of $T$.
\label{lemmaReflexEnough}
\end{lemma}

\section{Approximation Algorithm}
Given an orthogonal terrain $T$, our algorithm computes a subset of $V(T)$ that can guard all vertices of $T$, and we prove that the output of our algorithm is at most twice the size of the optimal solution for OTGP$(V(T), V(T))$.

By Lemmas~\ref{lemmaReflexSub} and~\ref{lemmaReflexEnough}, our problem can be reduced to OTGP$(V_{r}(T), V_{c}(T))$ \cite{Katz}. Let $G^*\subseteq{V_{r}(T)}$ be an optimal solution for OTGP$(V_{r}(T), V_{c}(T))$, $G^*$ can guard all convex vertices. So, of course, $G^*$ can guard all left convex vertices, i.e., $G^*$ has at least the same size as the optimal solution for OTGP$(V_{r}(T), V_{lc}(T))$. The same is true for $V_{rc}(T)$, the right convex vertices. 

Our algorithm first computes the optimal solutions for OTGP$(V_{r}(T), V_{lc}(T))$ and OTGP$(V_{r}(T), V_{rc}(T))$, then take the union of these two sets. Our solution can guard all convex vertices, and has the size at most twice as $G^*$, which means it is a 2-approximation.

In the following sections, we will present a sweep line algorithm that computes the optimal solution for OTGP$(V_{r}(T), V_{lc}(T))$. The right convex vertices part is symmetric. 

\subsection{Data Structures}
Our algorithm sweeps the terrain from right to left and put each left convex vertex $u$ into an associated list of a unique reflex vertex $v$, called $L(v)$. When the algorithm terminates, the set of all vertices with non-empty associated lists forms the solution, with each reflex vertex responsible to guard all left convex vertices in its associated list. In addition to the associated lists, following data structures are used:

(1) A modified stack $\mathcal{MS}$ to store a set of all reflex vertices each with a non-empty associated list and can potentially guard more left convex vertices beyond the sweep line. In addition to the standard stack operations (Top, Pop, Push), this modified data structure also supports deletion from any place in the stack given a pointer to that element. Along with each vertex in $\mathcal{MS}$, we also dynamically maintain its obstacle which defines the unique shadow ray with respect to the current sweep line.

(2) A heap, $\mathcal{H}$, to maintain the interior intersections of shadow rays of vertices adjacent in $\mathcal{MS}$.

(3) An event queue $\mathcal{EQ}$ consists of two components, a list $\mathcal{EQ}_T$ to keep all vertices of $T$, and a pointer $\mathcal{EQ}_I$ for $\mathcal{H}$. Next event is the rightmost vertex/intersection from $\mathcal{EQ}_T$ and $\mathcal{EQ}_I$. After handling an event, we delete it from the corresponding component of the queue.

(4) A standard stack, $\mathcal{UHS}$, to store the upper hull used for computing right horizons $R(v)$. 

For each vertex $v$ in $\mathcal{MS}$, we keep two pointers for the shadow ray intersections with its two neighbors. Pointers corresponding to missing neighbors/intersections are set to null. Symmetrically, for each intersection in $\mathcal{H}$, we use two pointers to reach the origins of the corresponding shadow rays in $\mathcal{MS}$.

\subsection{Computing Right Horizons}
To compute $R(v)$, the rightmost vertex visible from a left convex vertex $v$, we use the sweep line algorithm for computing the upper hull of a point set.

\begin{lemma}
Let $v$ be a left convex vertex. If $v$ is the rightmost vertex on terrain $T$, $R(v) = U(v)$. Otherwise, $R(v)$ is the vertex right next to $v$ on the upper hull of all vertices to the right of $v$ together with $v$. 
\end{lemma}
\begin{proof}
If $v$ is the rightmost vertex, it is easy to see that $U(v)$ is the rightmost reflex vertex that can see $v$, i.e., $R(v) = U(v)$.
Otherwise, $v$ is always on the upper hull of the considered vertices since it is the leftmost one. There must be some vertex to the right of $v$ on the upper hull, because the rightmost vertex is always on the upper hull. Let $p$ be the vertex next to $v$ on the upper hull, so $\overline{vp}$ is nowhere below the terrain, i.e., $p$ can see $v$. For any vertex $q$ to the right of $p$, as the property of upper hull, we have $p$ higher than $\overline{qv}$, which means $q$ cannot see $v$. So $R(v) = p$.
\end{proof}

With the upper hull of the swept vertices maintained in $\mathcal{UHS}$, $R(v)$ of a vertex $v$ on the sweep line can be found in constant time. Since $T$ is x-monotone, $\mathcal{UHS}$ can be maintained in linear time.

\subsection{Sweep Line Algorithm}
Our algorithm which sweeps the terrain from right to left is depicted below. Handling of each event consists of updates on the relevant data structures, described below after Observation \ref{limitedvisibility} which motivates the first step in handling a right reflex vertex event.
\begin{algorithm}[h]
\caption{\textsc{Terrain\textendash Sweeping}}
\label{Sweeping}
\begin{algorithmic}[1]
\State Initialize $\mathcal{H}$, $\mathcal{MS}$ and all $L(v)$ to be empty
\State Initialize $\mathcal{EQ}$ using $T$ and $\mathcal{H}$
\While{$\mathcal{EQ}_T\neq \emptyset$}
\State Let $v$ be next event in $\mathcal{E}$
\If{$v\in V(T)$}
\State Update $\mathcal{UHS}$ \label{opupdate}
\State Handle the vertex $v$ \label{ophandle}
\Else
\State Handle the intersection $v$
\EndIf
\EndWhile
\State Return $\lbrace{g\mid L(g)\neq\emptyset}\rbrace$
\end{algorithmic}
\label{algosweep}
\end{algorithm}
\begin{obs}
A right reflex vertex can see at least one left convex vertex which is right below it, and at most two left convex vertices.
\label{limitedvisibility}
\end{obs}

\begin{enumerate}
	\item{Left convex vertex $v$}:
	
		(i) Repeatedly Pop($\mathcal{MS}$), until Top($\mathcal{MS}$) can see $v$ or Top($\mathcal{MS}$) is to the right of $R(v)$.
		
		(ii) If Top($\mathcal{MS}$) sees $v$, add $v$ to $L$(Top($\mathcal{MS}$)). Otherwise, Push $R(v)$ to $\mathcal{MS}$, add $v$ to $L(R(v))$, and set $obs(R(v))$ be $v$, see Figure~\ref{figdummyshadowray}. $\overrightarrow{R(v)v}$ is called a \emph{dummy shadow ray}.	
	
	\item{Right convex vertex $v$}: the only update is to $\mathcal{UHS}$ (in Line~\ref{opupdate}), so nothing to be done in Line~\ref{ophandle}.
	
	\item{Left reflex vertex $v$}: 
	
		(i) Repeatedly Pop($\mathcal{MS}$) until Top($\mathcal{MS}$) cannot see $v$. Push back the last popped vertex that can see $v$, and update its obstacle to be $v$, see Figure~\ref{figleftreflex}. 
		
		(ii) Whenever deleting a vertex from $\mathcal{MS}$, remove its corresponding intersections from $\mathcal{H}$. For the vertex that is pushed to $\mathcal{MS}$, insert the shadow ray intersection with its neighbor to $\mathcal{H}$ and set the corresponding pointers.

	\item{Right reflex vertex $v$}: 
	
		(i) Let $u$=Top($\mathcal{MS}$). Iteratively Pop($\mathcal{MS}$) if Top($\mathcal{MS}$) is lower than $v$. If $u$ is lower than $v$ and there is only one vertex $p$ in $L(u)$, delete $p$ from $L(u)$, add $p$ to $L(v)$, and push $v$ to $\mathcal{MS}$, see Figure~\ref{figrightreflex}. To correctly compute the intersections introduced by the new vertex $v$ in $\mathcal{MS}$, we set $obs(v)$ one step ahead to be the vertex who shares the same horizontal edge with $v$.

		(ii) Delete all vertices in $\mathcal{MS}$ that can see $v$ except for the rightmost one.
		 
		(iii) Update intersections in $\mathcal{H}$ as in 3(ii).
		
		\item{Intersection $v$}: 
		
		(i) If intersection $v$ is above terrain $T$, delete all vertices from $\mathcal{MS}$, whose shadow rays are incident in $v$, except for the rightmost one, see Figure~\ref{figintersection}.
		
		(ii) Update the intersections and pointers as in 3(ii)
\end{enumerate}

\begin{figure}[h]
	\centering
	\begin{subfigure}[b]{0.49\linewidth}
		\includegraphics[width=\linewidth]{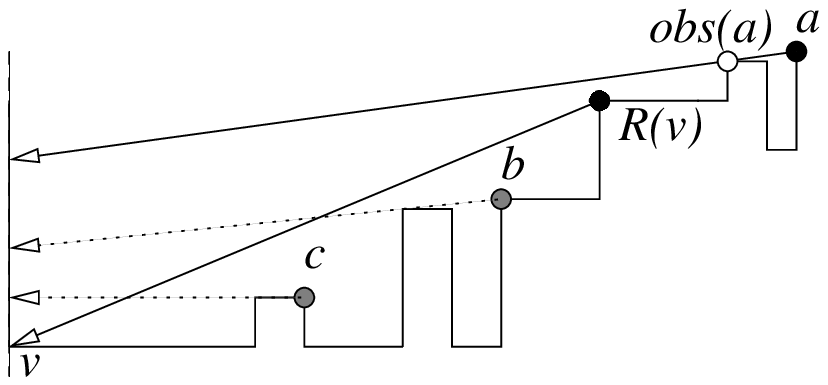}
		\caption{}
		\label{figdummyshadowray}
	\end{subfigure}
	\hfill
	\begin{subfigure}[b]{0.49\linewidth}
		\includegraphics[width=\linewidth]{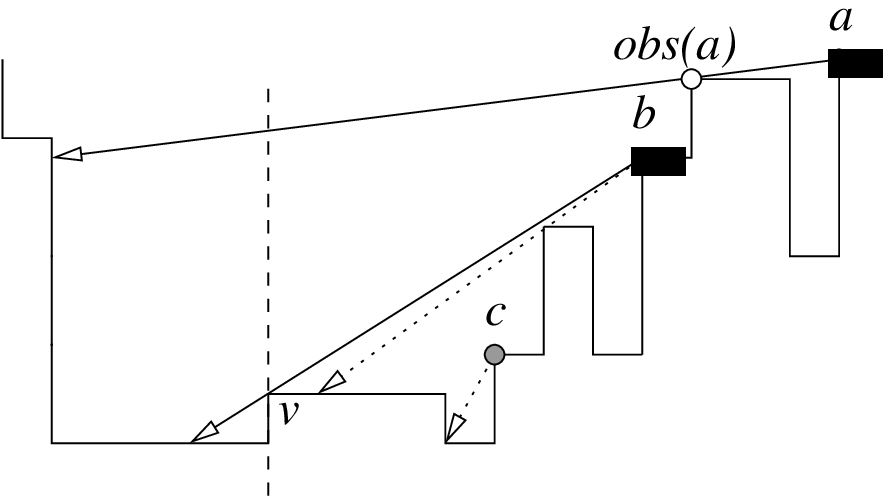}
		\caption{}
		\label{figleftreflex}
	\end{subfigure}
	\hfill
	\begin{subfigure}[b]{0.49\linewidth}
		\includegraphics[width=\linewidth]{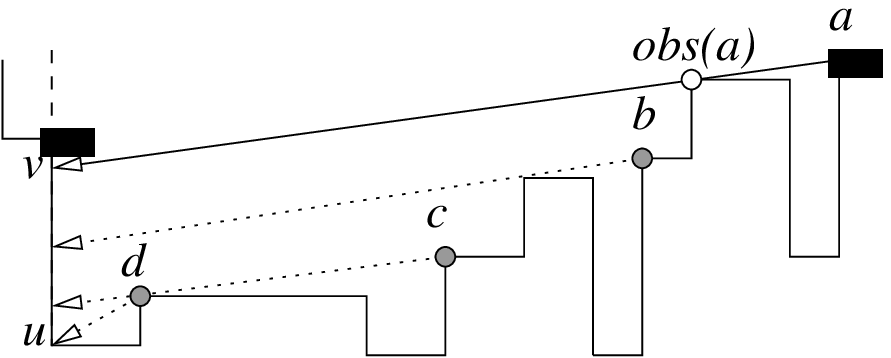}
		\caption{}
		\label{figrightreflex}
	\end{subfigure}
		\hfill
	\begin{subfigure}[b]{0.49\linewidth}
		\includegraphics[width=\linewidth]{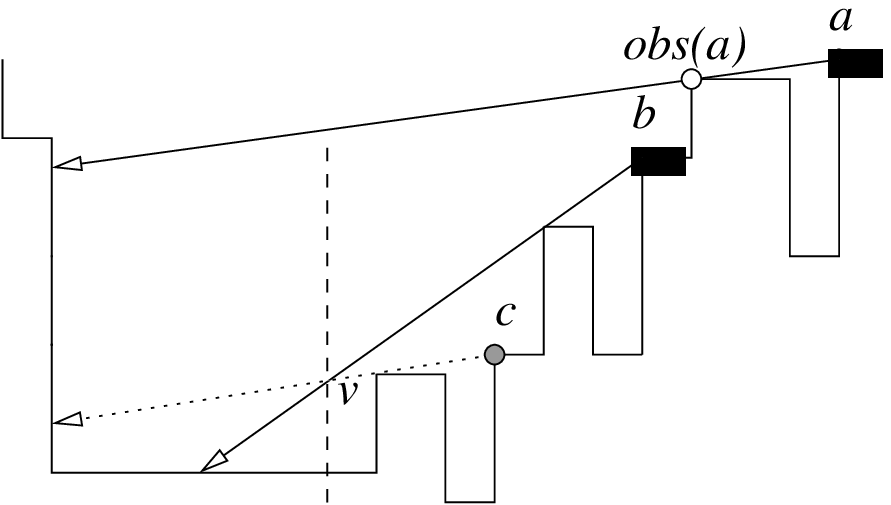}
		\caption{}
		\label{figintersection}
	\end{subfigure}
	\caption{(a) $v \in V_{lc}(T)$: remove reflex vertices from $\mathcal{MS}$ that are to the left of $R(v)$ and cannot see $v$, and add dummy shadow ray. (b) $v \in V_{lr}(T)$: remove all vertices from $\mathcal{MS}$ that can see $v$ except the rightmost one. (c) $v \in V_{rr}(T)$: delete all vertices that are lower than $v$. If $L(d)$ contains only one vertex, push $v$. (d) Intersection $v$: delete all vertices whose shadow rays are incident in $v$ except the rightmost one. }
\end{figure}

\subsection{Correctness}
We say a stack satisfies \emph{left to right order} if the vertices in the stack from top to bottom are ordered from left to right on the terrain. We say a stack satisfies \emph{lower to higher order} if the vertices in the stack from top to bottom are ordered from lower to higher on the terrain. If the stack satisfies both left to right order and lower to higher order, we say the stack is \emph{in order}.

\begin{lemma}
$\mathcal{MS}$ is always in order throughout Algorithm~\ref{algosweep}. The slope of each shadow ray is never negative, i.e., for each vertex $u$ in $\mathcal{MS}$, $obs(u)$ is never higher than $u$.
\end{lemma}
\begin{proof}
(By induction.) Initially, $\mathcal{MS}$ is empty. So the base case is trivial. Suppose before sweeping to event $v$, $\mathcal{MS}$ is in order and $obs(u)$ is no higher than $u$ for each $u$ in $\mathcal{MS}$.

(1) If $v$ is a left convex vertex, there are two cases. (i) If there exists any vertex in $\mathcal{MS}$ that can see $v$, we only pop vertices from $\mathcal{MS}$, so it is still in order. (ii) If no vertex can see $v$, all vertices to the left of $R(v)$ are deleted, and $R(v)$ is pushed into $\mathcal{MS}$. So, the left to right order is maintained. Next, we need to prove that all remaining vertices in $\mathcal{MS}$ are no lower than $R(v)$. Suppose there exists such vertex $u$ in $\mathcal{MS}$ that is lower than $R(v)$. Then a walk from $R(v)$ to $u$ on the terrain must go down a right reflex vertex $w$ that is higher than $u$. It is easy to see that $u$ cannot see any left convex vertex between $R(v)$ and $w$, so it must be pushed before sweeping to $w$. However, when the sweep line arrives at $w$, $u$ is deleted from $\mathcal{MS}$ as it is lower than $w$ as we will prove shortly. It is a contradiction. So all the other vertices are higher than $R(v)$. Also it is easy to see that the slope of dummy shadow ray $\overrightarrow{R(v)v}$ is positive.
	
(2) If $v$ is a right convex vertex, the only operation is updating the upper hull, $\mathcal{MS}$ remains the same.

(3) If $v$ is a left reflex vertex, we delete some vertices from $\mathcal{MS}$ and update the obstacle of a vertex $p$ to be $v$. As $p$ can see $v$ and $v$ is a left reflex vertex, $v$ is no lower than $p$.
	
(4) If $v$ is a right reflex vertex, as $\mathcal{MS}$ is in order by induction, step 4(i) ensures all the vertices that are lower than $v$ are deleted. Then if we push $v$ back to $\mathcal{MS}$, it is in order. Our newly introduced shadow ray is horizontal and the remaining operations are deletions.
	
(5) If $v$ is an intersection, we only delete some vertices from $\mathcal{MS}$.

Other than these events, $\mathcal{MS}$ will not change. So we can conclude that $\mathcal{MS}$ is always in order and the slope of each shadow ray is never negative.
\end{proof}

As a result of this lemma along with the definition of shadow ray, we can see that the obstacles can only be left reflex vertices except for the dummy shadow rays.

\begin{lemma}
For each vertex $v$ in $\mathcal{MS}$, $v$ and $obs(v)$ correctly define $sr_w(v)$ where $w$ is the current event. Shadow rays of vertices in $\mathcal{MS}$ have no pairwise interior intersections to the right of $w$, and are ordered from lower to higher corresponding to the order of their origins in $\mathcal{MS}$, with the lowest shadow ray corresponding to Top($\mathcal{MS}$).
\end{lemma}
\begin{proof}
(By induction.) Initially, $\mathcal{MS}$ is empty, hence the base case is trivial.
Suppose before dealing with event $w$ the claim holds. 

(1) $w$ is a left convex vertex: The shadow rays remain the same if the lowest shadow ray can see $w$. Otherwise, the vertices lower than $R(w)$ are deleted, and $R(w)$ is pushed into $\mathcal{MS}$ with $obs(R(w))=w$. Let $u$ be the vertex next to Top($\mathcal{MS}$). By definition, $sr_w(u)$ should be no lower than $R(w)$. $u$ cannot see $w$ as it is to the right of $R(w)$, i.e., $sr_w(u)$ is higher than $w$. So, $sr_w(u)$ is higher than $\overrightarrow{R(w)w}$ and no interior intersection is introduced to the right of $w$. The lemma holds.

(2) $w$ is a right convex vertex: The shadow rays remain the same.

(3) $w$ is a left reflex vertex: It is the only place we may need to update obstacles to keep the shadow rays correct. As the shadow rays are in order from lower to higher, all the vertices that can see $w$ are near the top of $\mathcal{MS}$ and are consecutive. So, our algorithm correctly finds all shadow rays that need to be updated. We delete all of them except the highest shadow ray which correspond to the rightmost vertex $v$ in $\mathcal{MS}$ that is visible from $w$, then update $sr_w(v)$. Similar to the arguments in case (1), $sr_w(v)$ is lower than the shadow rays of all vertices in $\mathcal{MS}$.

(4) $w$ is a right reflex vertex: The only place to push a vertex to $\mathcal{MS}$ is the first step and it can only push $w$. Suppose $w$ is pushed into $\mathcal{MS}$. In the second step, if $w$ is higher than the shadow ray of $p$ to its right in $\mathcal{MS}$, we will delete $w$ from $\mathcal{MS}$. Otherwise the shadow ray of $w$ is also lower than the shadow rays of all the other vertices in $\mathcal{MS}$.

(5) $w$ is an intersection: Under the induction assumption, the rightmost intersection appears between shadow rays of adjacent vertices in $\mathcal{MS}$. The way we maintain $\mathcal{H}$ ensures $w$ as the rightmost intersection. All shadow rays incident in $w$ are deleted except one, so $w$ disappears. 
\end{proof}

We say a point $p\in{V_{r}(T)}$ dominates point $q\in{V_{r}(T)}$, if $p$ can see every point $v\in{V_{lc}(T)}$ to the left of the sweep line that is visible by $q$.
\begin{lemma}
All vertices deleted from $\mathcal{MS}$ are either dominated by some vertex in $\mathcal{MS}$ at the end of current iteration, or cannot see any left convex vertex to the left of current sweep line.
\label{correctlemma}
\end{lemma}
\begin{proof}
Consider five types of event $v$:

(1) $v$ is a left convex vertex: We prove that all deleted vertices are dominated by the vertex whose associated list contains $v$ at the end of current iteration. Let this vertex be $p$. As $\mathcal{MS}$ is in order, any deleted vertex $u$ is to the left of $p$ and to the right of $v$. Suppose $u$ can see $q$ to the left of the sweep line. So we have $q<v<u<p$, $q$ can see $u$, and $v$ can see $p$. According to Lemma~\ref{lemmaOrder}, $q$ can see $p$; hence, $p$ dominates $u$.

(2) $v$ is a right convex vertex: No vertex is deleted.

(3) $v$ is a left reflex vertex: Let $p$ be the rightmost vertex in $\mathcal{MS}$ that can see $v$. We prove that all deleted vertices are dominated by $p$. Any deleted vertex $u$ must see $v$. Hence, $v<u<p$. Using a proof similar to case (1) and Lemma~\ref{lemmaOrder}, we conclude that $p$ dominates $u$.

(4) $v$ is a right reflex vertex: All vertices deleted in the first step are lower than $v$, so they cannot see any left convex vertex to the left of the sweep line. Similar to case (3), all vertices deleted in the second step are dominated by the rightmost one in $\mathcal{MS}$ that can see $v$.

(5) $v$ is an intersection: We prove that all deleted vertices are dominated by the rightmost vertex $p$ in $\mathcal{MS}$ whose shadow ray crosses $v$. As $\mathcal{MS}$ is in order, any deleted vertex $u$ is lower than and to the left of $p$. Suppose $u$ can see $q$ to the left of the sweep line, i.e., $q<v$. Segment $\overline{qu}$ is nowhere below the terrain $T$ and intersects segment $\overline{vp}$ in its interior. So $\overline{qp}$ is nowhere below the terrain $T$, which means $p$ can see $q$, see Figure~\ref{figintersectionlemma}.
\end{proof}
\begin{figure}[h]
\centering
\includegraphics[width=0.7\linewidth]{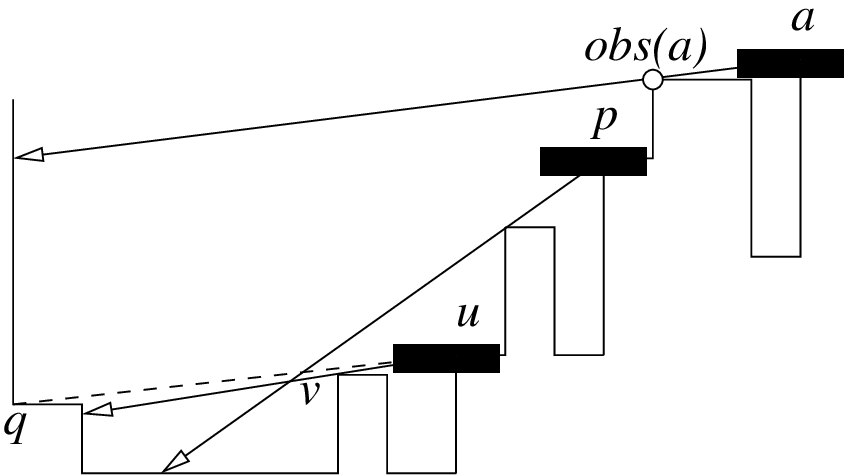}
\caption{intersection $v$: $u$ is dominated by $p$.}
\label{figintersectionlemma}
\end{figure}

Applying Lemma~\ref{correctlemma}, we can get the following corollary.
\begin{cor}
For any $v \in V_{lc}(T)$, if $v$ is seen by some vertex in $\mathcal{MS}$ before the sweep line reaches $v$, then $v$ is seen by some vertex in $\mathcal{MS}$ when the sweep line arrives at $v$.
\label{cor:see}
\end{cor}
Let our solution be set $G$, so we have for each $g\in{G}$, $L(g)$ is not empty. 
\begin{lemma}
For each left convex vertex $v$, there is a unique $g \in G$ such that $v \in L(g)$. 
\label{lem:outputsize}
\end{lemma}
\begin{proof}
Before the sweep reaches $v$, $v$ is not added to the list of any vertex. When the sweep line arrives at $v$, $v$ is added to some list $L(u)$. After that, the only operation that may change the list containing $v$ is the first step in handling right reflex vertex. If $L(u)$ contains some vertex other than $v$, $v$ will be in $L(u)$ till the end of the algorithm. If $L(u)$ contains only $v$, when $u$ is popped in the first step of handling right reflex vertex $w$, $v$ will be deleted from $L(u)$ and added to $L(w)$, then it will never change. In either case, when the algorithm terminates, there is a unique $g \in G$ such that $v \in L(g)$.
\end{proof}

Optimality of $G$ will be based on the following set definition also appears in \cite{Durocher}.
Let $F=\lbrace{v\vert{v}}$ is the first left convex vertex in $L(g)$, for each ${{g\in{G}}}\rbrace$. Observe that the sizes of the sets of $F$ and $G$ are the same, i.e., ${\vert{F}\vert}=\vert{G}\vert$. Moreover, for any vertex $v\in F$, we know that when the sweep line arrives at $v$, there is no vertex in $\mathcal{MS}$ that can see $v$.
\begin{lemma}
For any two vertices $u,v\in{F}$, there are no reflex vertices that can see both of them.
\end{lemma}
\begin{proof} 
To prove by contradiction, suppose $w$ is a reflex vertex that can see both $u, v \in F$. Without loss of generality, let $u<v$, so we visit $v$ first. Then we prove that there exists some vertex in $\mathcal{MS}$ that can see $u$ before the sweep line reaches $u$.

case 1: $w$ is a left reflex vertex: We have $u<v<w$. By definition of $R(v)$, $w\leqslant{R(v)}$. Using Lemma~\ref{lemmaOrder}, $R(v)$ can see $u$. When we visit $v$, we add $R(v)$ to $\mathcal{MS}$.

case 2: $w$ is a right reflex vertex: $w$ should be $U(v)$. It is easy to see that $R(v)$ is Top($\mathcal{MS}$) when the sweep line arrives at $w$. If $R(v)$ is higher than $U(v)$, $\overline{R(v)u}$ is nowhere below the terrain $T$, $R(v)$ can see $u$. Otherwise, $R(v)$ is popped in the first step as it is lower than $w$, and $w$ is pushed into $\mathcal{MS}$ as $L(R(v))$ contains only $v$.

In either case there exists some vertex in $\mathcal{MS}$ that can see $u$ before the sweep line reaches $u$. By Corollary~\ref{cor:see}, there exists some vertex in $\mathcal{MS}$ when the sweep line arrives at $u$, which contradicts that $u \in F$.
\end{proof}

Lemma \ref{lem:outputsize} implies that the optimal solution of OTGP$(V_r(T), V_{lc}(T))$ has at least ${\vert{F}\vert}$ reflex vertices. Our solution can see all left convex vertices and has size ${\vert{G}\vert}=\vert{F}\vert$. So we have the following result.
\begin{lemma}
Algorithm~\ref{algosweep} computes the optimal solution for OTGP$(V_r(T), V_{lc}(T))$.
\end{lemma}
Symmetrically we can compute the optimal solution for OTGP$(V_r(T), V_{rc}(T))$, leading to a 2-approximation algorithm for the OTGP$(V(T), V(T))$.

\subsection{Running Time}
Let $k$ be the size of $\mathcal{MS}$, and $t$ be the number of vertices with non-empty lists outside $\mathcal{MS}$. It is easy to see that the summation of $k$ and $t$ never decreases and eventually it will be $m$, where $m$ is the output size. As the number of intersections of shadow rays of adjacent vertices in $\mathcal{MS}$ is less than $k$, the size of $\mathcal{H}$ is $O(m)$. Note that $t$ is increased by at least 1 when handling each intersection. Thus there are $O(m)$ intersection events. Then we analyse the running time associated with each data structure.

(1) $\mathcal{UHS}$. Maintenance of upper hull takes $O(n)$ total time.

(2) $\mathcal{MS}$. The running time is proportional to the cost of stack insertions and deletions. Each deleted vertex when handling right reflex vertex $v$ is lower than $v$ and all the other deleted vertices are dominated by some vertex in $\mathcal{MS}$ by Lemma~\ref{correctlemma}. So the deleted vertices cannot be inserted again in future iterations. Each operation takes constant time. The total running time is $O(n)$.

(3) $\mathcal{H}$ and $\mathcal{EQ}$. There are four cases. (i) Get the next event. If the next event is from $\mathcal{EQ}_T$, it takes constant time and there are $n$ such events, so it takes $O(n)$ time in total; if next event is from $\mathcal{EQ}_I$, it takes $O(\log m)$ time and there are $O(m)$ intersections, so it takes $O(m \log m)$ in total. (ii) Insert vertices into $\mathcal{MS}$. There are $O(n)$ insertions and constant number of new intersections with each insertion, so the time complexity is $O(n \log m)$ in total. (iii) Delete vertices from $\mathcal{MS}$. Similar to case (ii). (iv) Update obstacles. We need to update at most one obstacle at any left reflex vertex, along with two deletions and one insertion with $\mathcal{H}$. As there are $O(n)$ left reflex vertices, the total running time is $O(n \log m)$.

Overall, the running time is $O(n \log m)$.


{
\small
\bibliographystyle{abbrv}

}

\newpage
\section*{Appendix: Removing Input Restriction}
We show how to remove the restriction when both sides are horizontal, and it is easier if only one side is horizontal.

First we extend the terrain by adding two edges. Let the leftmost vertex be $u$, and the rightmost vertex be $v$. We add two vertical edges $\overline{uu'}$ and $\overline{vv'}$ with infinitesimal length to both of them. The newly added vertex is the upper endpoint of the new edge. Let the extended terrain be $T'$. We have $V(T')=V(T)\cup{\lbrace{u', v'}\rbrace}$.
\begin{figure}[h]
	\centering
	\begin{subfigure}[b]{0.45\linewidth}
		\includegraphics[width=\linewidth]{figure/fig7.eps}
		\caption{}
		\label{figfirsthorzontala}
	\end{subfigure}
	\begin{subfigure}[b]{0.45\linewidth}
		\includegraphics[width=\linewidth]{figure/fig8.eps}
		\caption{}
		\label{figfirsthorzontalb}
	\end{subfigure}
	\caption{(a)next vertex to $v$ is a convex vertex. (b)next vertex to $v$ is a reflex vertex. If $v'$ is a guard, we can replace $v'$ with $w$.}
\end{figure}
\begin{lemma}
The cardinality of the optimal solution for OTGP$(V(T), V(T))$ is the same as the cardinality of the optimal solution for OTGP$(V(T'), V(T'))$, and we can easily transform from the solution of the latter to the solution of the former.
\end{lemma}
\begin{proof}
Let $G$ be an optimal solution for OTGP$(V(T), V(T))$, $G'$ be an optimal solution for OTGP$(V(T'), V(T'))$. Suppose $g \in G$ can see $u$, $g$ can also see $u'$, so $u'$ is seen by $G$. Similarly, $v'$ is also seen by $G$. $G \subseteq {V(T)} \subset V(T')$, we have $G$ is a solution for OTGP$(V(T'), V(T'))$. $\vert{G'}\vert \leqslant{\vert{G}\vert}$.

If neither of $u'$ and $v'$ is in $G'$, then $G' \subseteq {V(T)}$ and $G'$ can see $V(T)$, so $G'$ is a solution for OTGP$(V(T), V(T))$. $\vert{G}\vert \leqslant{\vert{G'}\vert}$. If $v'\in G'$, there are two cases depending on the vertex next to $v$. If the vertex next to $v$ is a left convex vertex as in Figure~\ref{figfirsthorzontala}. $v'$ can only see $p$, $U(p)$ and $v$, so we can replace $v'$ with $U(p)$. It is easy to see that $U(p)$ is not in $G'$, otherwise we get a better solution than $G'$ for OTGP$(V(T'), V(T'))$, it is a contradiction. Similarly we can find a replacement for $v'$ when the vertex next to $v$ is a reflex vertex, see Figure~\ref{figfirsthorzontalb}. We can also find a replacement for $u'$ if $u'\in G'$ symmetrically . Suppose we get an optimal solution $G''$ for OTGP$(V(T'), V(T'))$ after replacements. It is easy to see that $G''\subseteq V(T)$ and $G''$ can see $V(T)$, so $G''$ is an solution for OTGP$(V(T), V(T))$. $\vert{G}\vert \leqslant{\vert{G''}\vert} = \vert{G'}\vert$.

Thus we have $\vert{G}\vert = \vert{G'}\vert$, and we also showed how to transform from $G'$ to $G$.

\end{proof}

\end{document}